\pdfoutput=1
\documentclass[11pt]{article}
\usepackage{subfig}
\usepackage{amssymb,amsmath,amsthm}
\usepackage{graphics,epsfig}
\usepackage{algorithmic}
\usepackage{algorithm}
\usepackage[margin=1 in]{geometry}

\usepackage{mathtools}
\usepackage{cuted}
\usepackage{hhline}

\usepackage{multirow}
\usepackage{array}

\newtheorem{lem}{Lemma}[section]

\newtheorem{thm}[lem]{Theorem}

\newtheorem{prop}[lem]{Proposition}
\newtheorem{rem}[lem]{Remark}
\newcommand{\keywords}{\par\noindent{\textbf{Key words: }}}

\newcommand{\cmark}{\ensuremath{\checkmark}}
\newcommand{\xmark}{\ensuremath{\times}}

\title{Constant-Factor Approximation for Non-Splitting Coflow Scheduling in Heterogeneous Parallel Networks via Iterative Relaxation}
\author{Chi-Yeh~Chen 
\\ Department of Computer Science and Information
Engineering, \\ National Cheng Kung University, \\
Taiwan, ROC. \\
chency@csie.ncku.edu.tw.}

\begin{document}

\maketitle
\begin{abstract}
Coflows capture the collective communication requirements of data-parallel applications, and scheduling them well is essential to the performance of modern data centers. We study non-splitting coflow scheduling in a heterogeneous parallel network whose $m$ network cores may be Electronic Packet Switches (EPS), not-all-stop Optical Circuit Switches (OCS), or all-stop OCS, where every flow must be transmitted entirely through a single core. Previous approximation algorithms for this setting have guarantees that grow linearly with the number of cores $m$ or with the maximum flow degree $\tau$ of a port, and therefore become weak as data centers deploy more cores. The source of this dependence is the rounding step: rounding each flow to the core holding its largest fractional value can inflate the load of a port by a factor of $m$. We replace this step with an iterative relaxation rounding scheme that handles the input-port and output-port capacity constraints simultaneously, and we prove that the resulting load of every port on every core is at most $\min\{\tau, 3\}$ times the optimal value of the linear programming relaxation. For minimizing the makespan, this yields approximation ratios of $\min\{\tau, 3\}$, $2\min\{\tau, 3\}$, and $4\min\{\tau, 3\}$ for pure EPS, pure not-all-stop OCS, and pure all-stop OCS networks, respectively. For minimizing the total weighted coflow completion time, we obtain approximation ratios of $8\min\{\tau, 4\}$, $16\min\{\tau, 4\}$, and $32\min\{\tau, 4\}$ for the same three architectures. In a hybrid network, the guarantee is governed by the weakest architecture present, so the overall ratios are at most $4\min\{\tau,3\}\le 12$ and $32\min\{\tau,4\}\le 128$, respectively. Unlike all previous guarantees for non-splitting coflow scheduling in heterogeneous parallel networks, these bounds are absolute constants that do not grow with the number of cores $m$ or with the flow degree $\tau$. We also show that the integrality gap of the underlying linear program is at least $2$.

\end{abstract}

\section{Introduction}\label{sec:introduction}
In data-parallel frameworks such as MapReduce~\cite{Dean2008}, Hadoop~\cite{Shvachko2010, borthakur2007hadoop}, Dryad~\cite{isard2007dryad}, and Spark~\cite{zaharia2010spark}, a computation stage produces a large amount of intermediate data that must be shuffled between two groups of machines before the next stage can begin. The coflow abstraction~\cite{Chowdhury2012} groups the parallel flows of such a communication stage into a single unit. Because the next stage cannot start until all of its input has arrived, the quantity that matters to the application is the completion time of the last flow of the coflow~\cite{Chowdhury2012, shafiee2018improved}, and a number of systems have shown that scheduling the network at the granularity of coflows yields substantial gains~\cite{Chowdhury2014, Chowdhury2015, Zhang2016, Agarwal2018}.

The switching fabric of data centers is also becoming increasingly heterogeneous. An Electronic Packet Switch (EPS) multiplexes packets of many flows on a link and allows flows to be preempted. An Optical Circuit Switch (OCS)~\cite{Zhang2021, Tan2021, Li2022} offers much higher bandwidth at lower energy cost, but each of its ports can maintain only one circuit at any moment, and every change of the circuit configuration costs a reconfiguration delay~\cite{Zhang2021}. Depending on how reconfiguration is carried out, an OCS is either all-stop, in which case all circuits pause during a reconfiguration, or not-all-stop, in which case only the ports being reconfigured are interrupted. A realistic data center therefore contains several network cores of different types and speeds, and a coflow scheduler must decide both which core carries each flow and when the flow is transmitted.

We consider a set of coflows, each described by a demand matrix $D^{(k)}=\left(d_{ijk}\right)_{i,j=1}^{N}$, that must be transmitted through $m$ heterogeneous network cores, where core $q$ has link rate $r_{q}$. Flows are not allowed to be split: every flow is carried entirely by one core, although different flows of the same coflow may use different cores. Splitting a flow across fabrics causes packet reordering and other overhead that is undesirable in practice~\cite{chen2025nonsplitting}. We study two objectives: the makespan $\max_{k}\{C_{k}\}$ and the total weighted coflow completion time $\sum_{k} w_{k}C_{k}$, where $C_{k}$ and $w_{k}$ denote the completion time and the weight of coflow $k$.

\begin{table*}[htbp]
\centering
\caption{Comparison of related works (adapted and extended from~\cite{10411848, 11200997, chen2025nonsplitting}).}
\label{tab:network_comparison}
\begin{tabular}{|l|c|c|c|c|c|c|}
\hline
\multicolumn{1}{|c|}{\multirow{2}{*}{Works}} & \multirow{2}{*}{EPS} & \multicolumn{2}{c|}{OCS} & No. of & Performance & Ratio independent \\ \cline{3-4}
\multicolumn{1}{|c|}{} & & \textit{all-stop} & \textit{not-all-stop} & cores & guarantee & of $m$ and $\tau$ \\ \hline
\cite{chowdhury2011managing, Chowdhury2014, Chowdhury2015, Zhang2016, zhao2015rapier} & \cmark & \xmark & \xmark & 1 & \xmark & -- \\ \hline
\cite{Qiu2015, khuller2016brief, shafiee2018improved, Li2016} & \cmark & \xmark & \xmark & 1 & \cmark & -- \\ \hline
\cite{Xu2018} & \xmark & \cmark & \xmark & 1 & \xmark & -- \\ \hline
\cite{Zhang2019, Tan2021} & \xmark & \cmark & \xmark & 1 & \cmark & -- \\ \hline
\cite{huang2016, Zhang2021} & \xmark & \xmark & \cmark & 1 & \cmark & -- \\ \hline
\cite{Li2022, 10217994} & \cmark & \xmark & \cmark & 2 & \xmark & -- \\ \hline
\cite{Liu2015} & \cmark & \cmark & \xmark & 2 & \xmark & -- \\ \hline
\cite{10411848} & \cmark & \cmark & \xmark & 2 & \cmark & -- \\ \hline
\cite{11200997} & \cmark & \cmark & \cmark & 2 & \cmark & -- \\ \hline
\cite{CHEN2023104752,chen2023efficient1, chen2023efficient2,Huang2020} & \cmark  & \xmark & \xmark & $m$ & \cmark & -- \\ \hline
\cite{chen2025nonsplitting} & \cmark & \cmark & \cmark & $m$ & \cmark & \xmark \\ \hline
\textbf{This work} & \cmark & \cmark & \cmark & $m$ & \cmark & \cmark \\ \hline
\end{tabular}
\end{table*}

\subsection{Related Work}
The coflow abstraction was introduced by Chowdhury and Stoica~\cite{Chowdhury2012}. Since then, coflow scheduling has been studied mainly in two families of networks, namely EPS-based and OCS-based networks, and more recently in networks that mix both. Table~\ref{tab:network_comparison} summarizes the works most closely related to ours; it extends the comparisons given in~\cite{10411848, 11200997, chen2025nonsplitting}.

\subsubsection{EPS-based networks}
For the total weighted coflow completion time on a single EPS core, a sequence of works~\cite{Qiu2015, ahmadi2020scheduling, khuller2016brief, shafiee2018improved} improved the approximation ratio from $\frac{64}{3}$ to $4$ when all release times are zero, and from $\frac{76}{3}$ to $5$ for arbitrary release times. On identical parallel networks, Chen~\cite{chen2023efficient1} obtained ratios of $6-\frac{2}{m}$ and $5-\frac{2}{m}$ for arbitrary and zero release times, and later extended the approach to coflows with precedence constraints, achieving an $O(\chi)$-approximation~\cite{chen2023efficient2}, where $\chi$ is the number of coflows on the longest path of the precedence graph. On heterogeneous parallel networks, Huang \textit{et al.}~\cite{Huang2020} gave an $O(m)$-approximation for the makespan of a single coflow, which Chen~\cite{CHEN2023104752} improved to $O\left(\frac{\log m}{\log \log m}\right)$; the same bound was also obtained for the total weighted completion time~\cite{CHEN2023104752}.

\subsubsection{OCS-based networks}
For all-stop OCS, Tan \textit{et al.}~\cite{Tan2021} designed Reco-Sin and obtained a $2$-approximation for a single coflow and an $8K$-approximation for $K$ coflows. Wang \textit{et al.}~\cite{10411848} gave approximation guarantees for a hybrid-switched network consisting of one EPS and one all-stop OCS. For not-all-stop OCS, Huang \textit{et al.}~\cite{huang2016} proposed Sunflow, which is a constant-factor approximation for a single coflow, and Zhang \textit{et al.}~\cite{Zhang2021} gave a $4$-approximation for multiple coflows. For hybrid electrical-optical fabrics, Li and Shen~\cite{Li2022} and Jiang \textit{et al.}~\cite{10217994} proposed heuristics, and the work in~\cite{11200997} provided guarantees for two cores including both types of OCS. All of the hybrid results with guarantees~\cite{10411848, 11200997} are restricted to $m=2$ cores and rely on splitting flows across the fabrics.

\subsubsection{Non-splitting scheduling on heterogeneous parallel networks}
The work most closely related to ours is~\cite{chen2025nonsplitting}, which gave the first approximation algorithms with provable guarantees for non-splitting coflow scheduling in a hybrid network with an arbitrary number $m\ge 2$ of EPS, not-all-stop OCS, and all-stop OCS cores. The algorithms in~\cite{chen2025nonsplitting} first solve a linear programming relaxation of the flow-to-core assignment and then round each flow to the core with the largest fractional value. Let $\tau$ be the maximum number of flows incident to a single port. For the makespan, the resulting ratios are $\min\{\tau, m\}$, $2\min\{\tau, m\}$, and $2\min\{2\tau-1, m+\tau-1\}$ in pure EPS, pure not-all-stop OCS, and pure all-stop OCS networks, respectively; for the total weighted coflow completion time, they are $8\min\{\tau, 2m\}$, $16\min\{\tau, 2m\}$, and $16\min\{2\tau-1, 2m+\tau-1\}$~\cite{chen2025nonsplitting}. In a hybrid network, the guarantee is determined by the weakest architecture present.

Every one of these ratios grows linearly in $m$ or $\tau$, so they lose their meaning precisely when the number of cores is large. The reason is inherent to the rounding step. The largest of the $m$ fractional values of a flow is only guaranteed to be at least $\frac{1}{m}$, and if many flows of a port are spread evenly over all cores by the linear program and then rounded onto the same core, the load of that port grows by a factor of $m$. Rounding methods based on bipartite matching~\cite{shmoys1993approximation} control the load of one side, but every assignment variable belongs to both an input-port constraint and an output-port constraint, so matching on the input side alone gives no control on the output side~\cite{chen2025nonsplitting}. \textbf{The goal of this paper is to remove the factors $m$ and $\tau$ from the approximation ratios of~\cite{chen2025nonsplitting}.}

\subsection{Our Contributions}
We keep the linear programming relaxations of~\cite{chen2025nonsplitting} and replace their rounding step by an iterative relaxation scheme in the spirit of Lau, Ravi, and Singh~\cite{LauRaviSingh2011}. Our contributions are as follows.
\begin{itemize}
\item \textbf{Iterative relaxation rounding for two families of capacity constraints.} In the assignment linear program, each variable appears in exactly one assignment constraint and in exactly two capacity constraints, one for its input port and one for its output port. We show that every extreme point of the residual linear program either contains an integral variable or contains a capacity constraint whose total fractional deficiency is at most $2$, and that dropping such a constraint costs at most twice the largest transmission time. Consequently, the load of every port on every core after rounding is at most $\min\{\tau, 3\}T^{*}$, where $T^{*}$ is the smallest value for which the relaxation is feasible. The previous bound was $\min\{\tau, m\}T^{*}$~\cite{chen2025nonsplitting}.

\item \textbf{Makespan.} We obtain approximation ratios of
\begin{itemize}
\item $\min\{\tau, 3\}$ in a pure EPS network,
\item $2\min\{\tau, 3\}$ in a pure not-all-stop OCS network, and
\item $4\min\{\tau, 3\}$ in a pure all-stop OCS network.
\end{itemize}
For all-stop OCS, we also remove the additive dependence on $\tau$ caused by the reconfiguration delays by bounding the number of reconfigurations at a port through the port load itself.

\item \textbf{Total weighted coflow completion time.} Applying the new rounding inside each time interval of the interval-indexed relaxation of~\cite{chen2025nonsplitting}, we obtain approximation ratios of
\begin{itemize}
\item $8\min\{\tau, 4\}$ in a pure EPS network,
\item $16\min\{\tau, 4\}$ in a pure not-all-stop OCS network, and
\item $32\min\{\tau, 4\}$ in a pure all-stop OCS network.
\end{itemize}

\item \textbf{Hybrid networks and limits of the relaxation.} In a hybrid network, the overall guarantee equals that of the weakest architecture present; hence the makespan ratio is at most $12$ and the weighted completion time ratio is at most $128$, independently of $m$ and $\tau$. We further show that the integrality gap of the makespan relaxation is at least $2$, so any analysis based on this relaxation cannot bound the port loads by less than $2T^{*}$.
\end{itemize}

\subsection{Organization}
Section~\ref{sec:Preliminaries} describes the model and the notation. Section~\ref{sec:Makespan} presents the iterative relaxation rounding and the algorithm for minimizing the makespan. Section~\ref{sec:Weighted} extends the approach to the total weighted coflow completion time. Section~\ref{sec:Conclusion} concludes the paper. For completeness, the Birkhoff-von Neumann based schedule for a single EPS core is reviewed in the appendix.

\section{Notation and Preliminaries}\label{sec:Preliminaries}
We adopt the network and traffic model of~\cite{chen2025nonsplitting}; Table~\ref{tab:notations} lists the notation used throughout the paper.

\begin{table}[ht]
\caption{Notation and terminology.}
    \centering
        \begin{tabular}{||c|p{4.6in}||}
    \hline
     $N$      & Number of input ports (and of output ports) of each core.         \\
    \hline
		$\mathcal{M}$, $m$  & Set of network cores and its size.\\
    \hline
		$\mathcal{M}_{e}$, $\mathcal{M}_{o}^{1}$, $\mathcal{M}_{o}^{2}$  & Sets of EPS cores, all-stop OCS cores, and not-all-stop OCS cores.\\
    \hline
     $\mathcal{I}, \mathcal{J}$ & Sets of source servers (input ports) and destination servers (output ports).         \\
    \hline
     $\mathcal{K}$ & Set of coflows.         \\
    \hline
     $n$      & Number of flows.         \\
    \hline
     $D^{(k)}$, $d_{ijk}$     & Demand matrix of coflow $k$; size of the flow of coflow $k$ from input $i$ to output $j$. \\
    \hline
		 $\mathcal{F}$, $\mathcal{F}_{i}$, $\mathcal{F}_{j}$ & Set of all flows; flows with source $i$; flows with destination $j$. \\
		\hline
		$r_{q}$, $\delta_{q}$	& Link rate of core $q$; reconfiguration delay of an OCS core $q$. \\
		\hline
		$p_{ijkq}$	& Transmission time of flow $(i, j, k)$ on core $q$. \\
		\hline
		$\tau$  & Maximum flow degree of a port, $\max\left\{\max_{i} |\mathcal{F}_{i}|,\max_{j} |\mathcal{F}_{j}|\right\}$. \\
    \hline
		 $r_k$, $w_{k}$, $C_k$     & Release time, weight, and completion time of coflow $k$.   \\
    \hline
     $C_{ijk}$ & Completion time of flow $(i, j, k)$. \\
    \hline
		$t_{l}$ & Interval end point $t_{l}=2^l$ for $l=0, 1, \ldots, L$. \\
		\hline
        \end{tabular}
    \label{tab:notations}
\end{table}

\subsection{Network Model}
The network consists of a set $\mathcal{M}$ of $m$ non-blocking $N \times N$ switches, called network cores, partitioned as $\mathcal{M} = \mathcal{M}_{e} \cup \mathcal{M}_{o}^{1} \cup \mathcal{M}_{o}^{2}$ with $m = m_{e} + m_{o}^{1} + m_{o}^{2}$, where $\mathcal{M}_{e}$ contains $m_{e}$ EPS cores, $\mathcal{M}_{o}^{1}$ contains $m_{o}^{1}$ all-stop OCS cores, and $\mathcal{M}_{o}^{2}$ contains $m_{o}^{2}$ not-all-stop OCS cores. Every core has the same $N$ input ports, attached to the source servers $\mathcal{I}=\{1,\ldots,N\}$, and the same $N$ output ports, attached to the destination servers $\mathcal{J}=\{1,\ldots,N\}$; thus each source-destination pair is connected by $m$ parallel links, one through each core, and each core can be viewed as a bipartite graph between $\mathcal{I}$ and $\mathcal{J}$. All links of core $q$ have rate $r_{q}$.

An EPS core multiplexes packets of different flows on a port. An OCS core is subject to the port constraint: at any time, each input port and each output port takes part in at most one circuit. Changing the circuit configuration of an OCS core $q$ takes a fixed reconfiguration delay $\delta_{q}$. In an all-stop OCS, all transmissions of the core are suspended during a reconfiguration; in a not-all-stop OCS, only the ports whose circuits change are interrupted.

\subsection{Traffic Model}
Coflow $k\in\mathcal{K}$ is described by a demand matrix $D^{(k)}=\left(d_{ijk}\right)_{i,j=1}^{N}$, where $d_{ijk}$ is the amount of data that coflow $k$ sends from input $i$ to output $j$. The triple $(i,j,k)$ with $d_{ijk}>0$ is called a flow, and as in~\cite{Qiu2015} all flows of a coflow are released together at time $r_{k}$. We write
\begin{itemize}
\item $\mathcal{F}=\left\{(i, j, k) \mid d_{ijk}>0, i \in \mathcal{I}, j \in \mathcal{J}, k \in \mathcal{K} \right\}$,
\item $\mathcal{F}_{i}=\left\{(i, j, k)\in\mathcal{F} \mid j \in \mathcal{J}, k \in \mathcal{K} \right\}$ for $i\in\mathcal{I}$, and
\item $\mathcal{F}_{j}=\left\{(i, j, k)\in\mathcal{F} \mid i \in \mathcal{I}, k \in \mathcal{K} \right\}$ for $j\in\mathcal{J}$,
\end{itemize}
and let $n=|\mathcal{F}|$. Flows are indivisible: every flow must be transmitted entirely through one core, whereas different flows of the same coflow may use different cores.

The transmission time of flow $(i,j,k)\in\mathcal{F}$ on core $q$ is
\begin{equation*}
p_{ijkq} =
\begin{cases}
\dfrac{d_{ijk}}{r_{q}}, & \text{if } q\in \mathcal{M}_{e},\\[2mm]
\dfrac{d_{ijk}}{r_{q}} + \delta_{q}, & \text{if } q\in \mathcal{M}_{o}^{1} \cup \mathcal{M}_{o}^{2},
\end{cases}
\end{equation*}
that is, on an OCS core every flow is charged one reconfiguration delay for setting up its circuit. Without loss of generality, we scale time so that $p_{ijkq}\geq 1$ for all flows and cores. For a set of flows assigned to core $q$, the \emph{load} of a port on core $q$ is the sum of the transmission times on core $q$ of the assigned flows incident to that port.

\subsection{Objectives}
For the makespan, all coflows arrive at a common time $t^{arr}$, and, with $t^{F}_{ijk}$ denoting the time at which flow $(i,j,k)$ finishes, we minimize
\begin{equation*}
C_{\max} = \max_{(i,j,k)\in \mathcal{F}} \left(t^{F}_{ijk} - t^{arr} \right).
\end{equation*}
For the second objective, coflow $k$ is released at time $r_{k}\ge 0$, its completion time $C_{k}$ is the completion time of its last flow, and we minimize $\sum_{k \in \mathcal{K}} w_{k}C_{k}$.

\section{Approximation Algorithm for Minimizing the Makespan}\label{sec:Makespan}
Like the algorithms of~\cite{chen2025nonsplitting}, our algorithm works in two stages. The first stage assigns every flow to a core by rounding a fractional solution of a linear program, and the second stage schedules, independently on each core, the flows that were assigned to it. The second stage is the same as in~\cite{chen2025nonsplitting}; the novelty lies in the rounding of the first stage.

\subsection{The Assignment Relaxation and the Bisection Procedure of~\cite{chen2025nonsplitting}}
For a flow $(i,j,k)\in\mathcal{F}$ and a core $q\in\mathcal{M}$, let $x_{ijkq}\in\{0,1\}$ indicate whether the flow is transmitted through core $q$. For a target makespan $T$, the following linear program, introduced in~\cite{chen2025nonsplitting}, relaxes the assignment problem:
\begin{subequations}\label{coflow:interval}
\begin{align}
& \text{LP($T$):}  &&      &   & \tag{\ref{coflow:interval}} \\
&  && \sum_{q\in \mathcal{M}} x_{ijkq} = 1, && \forall (i, j, k)\in \mathcal{F}, \label{interval:a} \\
&  && \sum_{(i, j, k) \in \mathcal{F}_{i}} p_{ijkq}\, x_{ijkq} \leq T, && \forall i\in \mathcal{I}, \forall q\in \mathcal{M}, \label{interval:b} \\
&  && \sum_{(i, j, k) \in \mathcal{F}_{j}} p_{ijkq}\, x_{ijkq} \leq T, && \forall j\in \mathcal{J}, \forall q\in \mathcal{M}, \label{interval:c} \\
&  && x_{ijkq} \geq 0, && \forall (i, j, k)\in \mathcal{F}, \forall q\in \mathcal{M}, \label{interval:d} \\
&  && x_{ijkq} = 0, && \forall (i, j, k)\in \mathcal{F}, \forall q\in \mathcal{M} \text{ with } p_{ijkq}>T. \label{interval:e}
\end{align}
\end{subequations}
Constraint~(\ref{interval:a}) assigns each flow to one core, constraints~(\ref{interval:b}) and~(\ref{interval:c}) limit the load of every input port and every output port on every core to $T$, and constraint~(\ref{interval:e}) forbids a core on which the flow alone would take longer than $T$. In any feasible non-splitting schedule with makespan $T$, the load of every port on every core is at most $T$ and every flow is carried by a core on which it takes at most $T$ time units, so the schedule induces an integral solution of LP($T$). Hence LP($T$) is feasible for $T=\mathrm{OPT}$, the optimal makespan.

Because constraint~(\ref{interval:e}) depends on $T$, the smallest feasible $T$ is not computed by a single linear program. Following~\cite{chen2025nonsplitting}, we use bisection (Algorithm~\ref{Alg1}). The initial upper bound $u$ is the largest port load obtained when every flow uses its slowest core, for which LP($u$) is clearly feasible, and the initial lower bound $l$ is the largest port demand divided by the total rate of all cores. In each step, the feasibility of LP($d$) is tested for the midpoint $d$, and the interval is halved accordingly. We denote by $T^{*}$ the value returned by the procedure. Since LP(OPT) is feasible, we have
\begin{equation}\label{eq:Tstar}
T^{*} \leq \mathrm{OPT}.
\end{equation}
A core $q$ with $\delta_{q}\geq T^{*}$ cannot carry any flow in LP($T^{*}$) and is discarded. Thus, we may assume $\delta_{q}<T^{*}$ for every remaining OCS core.

\begin{algorithm}
\caption{Bisection procedure~\cite{chen2025nonsplitting}}
    \begin{algorithmic}[1]
				\STATE $u\leftarrow \max\left\{\max_{i\in \mathcal{I}} \sum_{(i,j,k)\in\mathcal{F}_{i}}\max_{q}p_{ijkq},\ \max_{j\in \mathcal{J}} \sum_{(i,j,k)\in\mathcal{F}_{j}}\max_{q}p_{ijkq}\right\}$
				\STATE $\rho\leftarrow \max \left\{\max_{i}\sum_{(i,j,k)\in\mathcal{F}_{i}}d_{ijk},\ \max_{j}\sum_{(i,j,k)\in\mathcal{F}_{j}}d_{ijk}\right\}$
				\STATE $l\leftarrow \rho/\sum_{q\in\mathcal{M}} r_{q}$
				\WHILE{$u\neq l$}
							 \STATE $d\leftarrow\left\lfloor \frac{u+l}{2}\right\rfloor$
							 \IF{LP($d$) is feasible}
							     \STATE $u\leftarrow d$
							 \ELSE
							     \STATE $l\leftarrow d+1$
							 \ENDIF
				\ENDWHILE
				\STATE \textbf{return} $T^{*}\leftarrow u$
   \end{algorithmic}
\label{Alg1}
\end{algorithm}

In~\cite{chen2025nonsplitting}, a fractional solution of LP($T^{*}$) is rounded by sending every flow to the core $q^{*}=\arg\max_{q} x_{ijkq}$. Since $x_{ijkq^{*}}\geq \frac{1}{m}$, the load of a port grows by at most a factor of $m$, which gives the bound $\min\{\tau,m\}T^{*}$. The example in which the flows of a port are spread evenly over the $m$ cores and then all rounded onto one core shows that this analysis cannot be improved. We now describe a rounding procedure that avoids this loss.

\subsection{Iterative Relaxation Rounding}
It is convenient to describe the rounding for a general assignment polytope, since the same procedure will be used again in Section~\ref{sec:Weighted}. Let $\bar{\mathcal{F}}\subseteq\mathcal{F}$ be a set of flows, let $T>0$, and let $\mathcal{S}\subseteq \bar{\mathcal{F}}\times\mathcal{M}$ be a set of admissible flow-core pairs such that
\begin{equation}\label{eq:pmax}
p_{ijkq}\leq T \quad \text{for every } ((i,j,k),q)\in\mathcal{S}.
\end{equation}
We call a pair $e=((i,j,k),q)\in\mathcal{S}$ an \emph{edge}, write $p_{e}=p_{ijkq}$ and $x_{e}=x_{ijkq}$, and say that $e$ belongs to flow $(i,j,k)$. For every input port $i$ and core $q$, there is a capacity constraint $c=(i,q)$, and for every output port $j$ and core $q$, there is a capacity constraint $c=(j,q)$. An edge $((i,j,k),q)$ belongs to exactly two capacity constraints, namely $(i,q)$ and $(j,q)$. Let $\mathcal{C}$ be the set of all $2Nm$ capacity constraints. The polytope
\begin{equation}\label{eq:P}
\mathcal{P}(\bar{\mathcal{F}},\mathcal{S},T)=\left\{ \textbf{x}\in\mathbb{R}_{\geq 0}^{\mathcal{S}} \ \middle|\ \sum_{e\in\mathcal{S}:\, e \text{ belongs to } f} x_{e}=1\ \ \forall f\in\bar{\mathcal{F}};\ \ \sum_{e\in\mathcal{S}:\, e \in c} p_{e}x_{e}\leq T\ \ \forall c\in\mathcal{C}\right\}
\end{equation}
is exactly the feasible region of LP($T$) when $\bar{\mathcal{F}}=\mathcal{F}$ and $\mathcal{S}=\{((i,j,k),q)\mid p_{ijkq}\leq T\}$.

The procedure (Algorithm~\ref{AlgIR}) maintains a set $\mathcal{F}'$ of unassigned flows, a set $\mathcal{S}'$ of edges still available to them, a set $\mathcal{C}'\subseteq\mathcal{C}$ of capacity constraints still enforced, and, for each $c\in\mathcal{C}'$, a residual capacity $T_{c}$, which equals $T$ minus the load already placed on $c$. The residual linear program is
\begin{subequations}\label{RLP}
\begin{align}
& \text{RLP:} && \sum_{e\in\mathcal{S}':\, e \text{ belongs to } f} x_{e}=1, && \forall f\in\mathcal{F}', \label{RLP:a}\\
& && \sum_{e\in\mathcal{S}':\, e \in c} p_{e}x_{e}\leq T_{c}, && \forall c\in\mathcal{C}', \label{RLP:b}\\
& && x_{e}\geq 0, && \forall e\in\mathcal{S}'. \label{RLP:c}
\end{align}
\end{subequations}
For a constraint $c\in\mathcal{C}'$, let $E(c)=\{e\in\mathcal{S}' \mid e\in c\}$ and $d(c)=|E(c)|$. In each iteration, the algorithm computes an extreme point solution of RLP and deletes the edges whose value is zero. If some edge has value one, the corresponding flow is fixed to that core. Otherwise, the algorithm drops a capacity constraint whose total fractional deficiency $\sum_{e\in E(c)}(1-x_{e})$ is at most $2$; Lemma~\ref{lem:IR1} below shows that such a constraint always exists. This mirrors the iterative relaxation technique for the generalized assignment problem~\cite{LauRaviSingh2011}, with the difference that each of our variables appears in two capacity constraints instead of one.

\begin{algorithm}
\caption{Iterative relaxation rounding}
    \begin{algorithmic}[1]
		    \REQUIRE $\bar{\mathcal{F}}$, $\mathcal{S}$, and $T$ such that $\mathcal{P}(\bar{\mathcal{F}},\mathcal{S},T)\neq\emptyset$ and (\ref{eq:pmax}) holds
				\STATE $\mathcal{F}'\leftarrow\bar{\mathcal{F}}$, $\mathcal{S}'\leftarrow\mathcal{S}$, $\mathcal{C}'\leftarrow\mathcal{C}$, and $T_{c}\leftarrow T$ for every $c\in\mathcal{C}$
				\WHILE{$\mathcal{F}'\neq \emptyset$}
						\STATE Compute an extreme point solution $\textbf{x}$ of RLP
						\STATE $\mathcal{S}'\leftarrow\mathcal{S}'\setminus\{e\in\mathcal{S}' \mid x_{e}=0\}$
						\IF{there is an edge $e=((i,j,k),q)\in\mathcal{S}'$ with $x_{e}=1$}
								\STATE Assign flow $(i,j,k)$ to core $q$, i.e., $\hat{x}_{ijkq}\leftarrow 1$
								\STATE $\mathcal{F}'\leftarrow\mathcal{F}'\setminus\{(i,j,k)\}$ and remove from $\mathcal{S}'$ all edges of flow $(i,j,k)$
								\STATE For each $c\in\{(i,q),(j,q)\}\cap\mathcal{C}'$, set $T_{c}\leftarrow T_{c}-p_{ijkq}$
						\ELSE
								\STATE Find $c\in \mathcal{C}'$ with $d(c)\geq 1$ and $\sum_{e\in E(c)}(1-x_{e})\leq 2$
								\STATE $\mathcal{C}'\leftarrow\mathcal{C}'\setminus\{c\}$ \label{AlgIR:relax}
						\ENDIF
				\ENDWHILE
				\STATE \textbf{return} $\hat{\textbf{x}}$
   \end{algorithmic}
\label{AlgIR}
\end{algorithm}

We first show that RLP remains feasible, so that every iteration is well defined.

\begin{lem}\label{lem:feasible}
At the beginning of every iteration of Algorithm~\ref{AlgIR}, RLP is feasible and bounded.
\end{lem}
\begin{proof}
Initially, RLP is the system defining $\mathcal{P}(\bar{\mathcal{F}},\mathcal{S},T)$, which is nonempty by assumption. Suppose $\textbf{x}$ is the extreme point computed in some iteration. Deleting the edges with $x_{e}=0$ does not affect the feasibility of $\textbf{x}$. If an edge $e=((i,j,k),q)$ with $x_{e}=1$ is fixed, then all other edges of flow $(i,j,k)$ have value zero by~(\ref{RLP:a}), and the restriction of $\textbf{x}$ to the remaining edges satisfies~(\ref{RLP:a}) for the remaining flows and~(\ref{RLP:b}) with the reduced capacities $T_{c}-p_{e}$. If a capacity constraint is dropped, $\textbf{x}$ remains feasible. Hence RLP is feasible in the next iteration. It is bounded because~(\ref{RLP:a}) and~(\ref{RLP:c}) imply $0\leq x_{e}\leq 1$ for every edge.
\end{proof}

The next lemma is the heart of the analysis.

\begin{lem}\label{lem:IR1}
Let $\textbf{x}$ be an extreme point solution of RLP with $0<x_{e}<1$ for every $e\in\mathcal{S}'$, and suppose $\mathcal{F}'\neq\emptyset$. Then there exists a constraint $c\in\mathcal{C}'$ with $d(c)\geq 1$ and $\sum_{e\in E(c)}(1-x_{e})\leq 2$.
\end{lem}
\begin{proof}
Let $V=|\mathcal{S}'|$ be the number of variables and let $\mathcal{C}^{+}=\{c\in\mathcal{C}' \mid d(c)\geq 1\}$. Since no variable equals $1$ and the variables of each flow in $\mathcal{F}'$ sum to $1$, every flow in $\mathcal{F}'$ has at least two edges, and hence
\begin{equation}\label{eq:V1}
V\geq 2|\mathcal{F}'|.
\end{equation}
By the rank lemma~\cite{LauRaviSingh2011}, an extreme point of RLP is the unique solution of $V$ linearly independent constraints of RLP that are tight at $\textbf{x}$. No nonnegativity constraint is tight, and a capacity constraint with $d(c)=0$ has a zero left-hand side, so these $V$ constraints are assignment constraints~(\ref{RLP:a}) or capacity constraints in $\mathcal{C}^{+}$. Therefore,
\begin{equation}\label{eq:V2}
V\leq |\mathcal{F}'|+|\mathcal{C}^{+}|.
\end{equation}
If $\mathcal{C}^{+}=\emptyset$, then~(\ref{eq:V1}) and~(\ref{eq:V2}) give $2|\mathcal{F}'|\leq|\mathcal{F}'|$, which contradicts $\mathcal{F}'\neq\emptyset$; thus $\mathcal{C}^{+}\neq\emptyset$. Each edge belongs to at most two constraints of $\mathcal{C}^{+}$, and the values of the edges sum to $|\mathcal{F}'|$ by~(\ref{RLP:a}). Using~(\ref{eq:V2}), we obtain
\begin{equation*}
\sum_{c\in\mathcal{C}^{+}}\ \sum_{e\in E(c)}(1-x_{e})\ \leq\ 2\sum_{e\in\mathcal{S}'}(1-x_{e})\ =\ 2\left(V-|\mathcal{F}'|\right)\ \leq\ 2|\mathcal{C}^{+}|.
\end{equation*}
Hence the average of $\sum_{e\in E(c)}(1-x_{e})$ over $c\in\mathcal{C}^{+}$ is at most $2$, and some $c\in\mathcal{C}^{+}$ satisfies the claim.
\end{proof}

\begin{thm}\label{thm:IR}
Suppose $\mathcal{P}(\bar{\mathcal{F}},\mathcal{S},T)\neq\emptyset$ and~(\ref{eq:pmax}) holds. Then Algorithm~\ref{AlgIR} runs in polynomial time and assigns every flow of $\bar{\mathcal{F}}$ to exactly one core $q$ with $((i,j,k),q)\in\mathcal{S}$. Moreover, let $P=\max_{e\in\mathcal{S}}p_{e}$, so that $P\leq T$ by~(\ref{eq:pmax}). Then, for every core $q$, the load of every input port and every output port on $q$ is at most $T+2P\leq 3T$.
\end{thm}
\begin{proof}
By Lemma~\ref{lem:feasible}, every iteration computes an extreme point of a feasible and bounded linear program, which can be done in polynomial time. If no edge has value one, Lemma~\ref{lem:IR1} guarantees that a constraint can be dropped. Hence each iteration either assigns a flow or drops a constraint, and the algorithm stops after at most $|\bar{\mathcal{F}}|+2Nm$ iterations. Every flow is assigned through an edge of $\mathcal{S}$, and once assigned, it is removed from $\mathcal{F}'$, so it is assigned exactly once.

Fix a capacity constraint $c$, and let $\Lambda_{c}$ denote the total transmission time of the flows assigned to $c$ so far. As long as $c\in\mathcal{C}'$, we have $T_{c}=T-\Lambda_{c}$, and the feasibility of RLP implies $\Lambda_{c}\leq T$. If $c$ is never dropped, its final load is at most $T$. Otherwise, consider the iteration in which $c$ is dropped, with extreme point $\textbf{x}$ and current load $\Lambda_{c}$. Since edges are only deleted and never added, every flow that is assigned to $c$ afterwards is assigned through an edge of the set $E(c)$ at that moment. Hence the final load of $c$ is at most $\Lambda_{c}+\sum_{e\in E(c)}p_{e}$. We split this quantity into a fractional part and a deficiency part:
\begin{equation*}
\Lambda_{c}+\sum_{e\in E(c)}p_{e}
= \left(\Lambda_{c}+\sum_{e\in E(c)}p_{e}x_{e}\right)+\sum_{e\in E(c)}p_{e}(1-x_{e}).
\end{equation*}
Since $c\in\mathcal{C}'$ when $\textbf{x}$ is computed, constraint~(\ref{RLP:b}) gives $\sum_{e\in E(c)}p_{e}x_{e}\leq T_{c}=T-\Lambda_{c}$, so the first part is at most $T$. Since $1-x_{e}>0$ and $p_{e}\leq P$ for every $e\in E(c)$, and since $c$ satisfies the condition in Line~\ref{AlgIR:relax}, the second part is at most
\begin{equation*}
\sum_{e\in E(c)}p_{e}(1-x_{e})\leq P\sum_{e\in E(c)}(1-x_{e})\leq 2P.
\end{equation*}
Therefore the final load of $c$ is at most $T+2P$, which is at most $3T$ because $P\leq T$.
\end{proof}

The bound $T+2P$ separates the role of the capacity $T$ from that of the largest transmission time $P$. For LP($T^{*}$), the two coincide up to $P\leq T^{*}$, and we obtain $3T^{*}$. In Section~\ref{sec:Weighted}, however, the capacity of a group will be twice as large as the transmission times allowed in it, and the finer bound will give a better constant.

We now apply Algorithm~\ref{AlgIR} to LP($T^{*}$), that is, with $\bar{\mathcal{F}}=\mathcal{F}$, $T=T^{*}$, and $\mathcal{S}=\{((i,j,k),q)\mid p_{ijkq}\leq T^{*}\}$. Let $T_{i}^{r}$ and $T_{j}^{r}$ denote the load of input port $i$ and output port $j$ on a core after rounding; the index of the core is omitted since the bound is the same for all cores.

\begin{lem}\label{lem:portload}
After applying Algorithm~\ref{AlgIR} to LP($T^{*}$), we have $T_{i}^{r}, T_{j}^{r} \leq \min\left\{\tau, 3\right\} T^{*}$ for every core, every $i \in \mathcal{I}$, and every $j \in \mathcal{J}$.
\end{lem}
\begin{proof}
By the definition of the bisection procedure, LP($T^{*}$) is feasible, and~(\ref{eq:pmax}) holds with $T=T^{*}$ by the choice of $\mathcal{S}$. Theorem~\ref{thm:IR} gives $T_{i}^{r}, T_{j}^{r}\leq 3T^{*}$. In addition, at most $\tau$ flows are incident to any port, and each of them is assigned through an edge of $\mathcal{S}$, so its transmission time on the chosen core is at most $T^{*}$. Hence $T_{i}^{r}, T_{j}^{r}\leq \tau T^{*}$ as well.
\end{proof}

Lemma~\ref{lem:portload} replaces the bound $\min\{\tau, m\}T^{*}$ of~\cite{chen2025nonsplitting} by a bound that does not depend on $m$.

\subsection{Scheduling on Each Core}
After the assignment, each core schedules its own flows independently (Algorithm~\ref{AlgMS}). As in~\cite{chen2025nonsplitting}, an EPS core uses the Birkhoff-von Neumann decomposition~\cite{marcus1959diagonals, Qiu2015}, a not-all-stop OCS core uses Sunflow~\cite{huang2016}, and an all-stop OCS core uses Reco-Sin~\cite{Tan2021}.

\begin{algorithm}
\caption{Scheduling for minimizing the makespan}
    \begin{algorithmic}[1]
				\STATE Run Algorithm~\ref{Alg1} to obtain $T^{*}$
				\STATE Run Algorithm~\ref{AlgIR} with $\bar{\mathcal{F}}=\mathcal{F}$, $T=T^{*}$, and $\mathcal{S}=\{((i,j,k),q)\mid p_{ijkq}\leq T^{*}\}$ to obtain $\hat{\textbf{x}}$
				\FOR{each $q\in \mathcal{M}$ in parallel}
						\IF{$q\in\mathcal{M}_{e}$}
							\STATE Schedule the flows assigned to $q$ by the Birkhoff-von Neumann decomposition~\cite{marcus1959diagonals, Qiu2015}
						\ELSIF{$q\in\mathcal{M}_{o}^{2}$}
							\STATE Schedule the flows assigned to $q$ by Sunflow~\cite{huang2016}
						\ELSE
							\STATE Schedule the flows assigned to $q$ by Reco-Sin~\cite{Tan2021}
						\ENDIF
				\ENDFOR
   \end{algorithmic}
\label{AlgMS}
\end{algorithm}

\subsection{Analysis}
Let $T^{s}$ denote the makespan of the schedule produced by Algorithm~\ref{AlgMS}, and let $T^{s}_{q}$ be the time at which core $q$ finishes its flows, so that $T^{s}=\max_{q}T^{s}_{q}$. For a core $q$, let $L_{q}$ be the largest load of a port on $q$ after rounding. By Lemma~\ref{lem:portload},
\begin{equation}\label{eq:Lq}
L_{q}\leq \min\{\tau,3\}T^{*}\quad\text{for every } q\in\mathcal{M}.
\end{equation}

\begin{lem}\label{lem:EPS}
If $q\in\mathcal{M}_{e}$, then $T^{s}_{q}\leq \min\{\tau,3\}T^{*}$.
\end{lem}
\begin{proof}
By Theorem~\ref{BvN2} in the appendix, the Birkhoff-von Neumann based schedule finishes the flows of an EPS core within the largest port load of that core. Hence $T^{s}_{q}=L_{q}$, and the claim follows from~(\ref{eq:Lq}).
\end{proof}

\begin{lem}\label{lem:NAS}
If $q\in\mathcal{M}_{o}^{2}$, then $T^{s}_{q}\leq 2\min\{\tau,3\}T^{*}$.
\end{lem}
\begin{proof}
By the analysis of Sunflow~\cite{huang2016} (see also~\cite{chen2025nonsplitting}), the flow that finishes last on $q$, say from input $i$ to output $j$, finishes no later than the sum of the loads of $i$ and $j$ on $q$. Hence $T^{s}_{q}\leq T_{i}^{r}+T_{j}^{r}\leq 2L_{q}$, and the claim follows from~(\ref{eq:Lq}).
\end{proof}

For an all-stop OCS core, the schedule length consists of a transmission part and a reconfiguration part. For a core $q\in\mathcal{M}_{o}^{1}$, let $\rho_{q}$ be the largest amount of data that a single port of $q$ has to send or receive, and let $\Delta_{q}$ be the largest number of flows assigned to $q$ that are incident to a single port. The analysis of Reco-Sin~\cite{Tan2021}, as applied in~\cite{chen2025nonsplitting}, shows that
\begin{equation}\label{eq:recosin}
T^{s}_{q}\leq 2\left(\frac{\rho_{q}}{r_{q}}+\Delta_{q}\,\delta_{q}\right),
\end{equation}
that is, at most twice the data transmission time of the busiest port plus the reconfiguration delay multiplied by the largest number of flows at an input or output port. In~\cite{chen2025nonsplitting}, the second term was bounded by $\tau\delta_{q}<\tau T^{*}$, which introduces an additive dependence on $\tau$. The following lemma shows that both terms are dominated by the port loads.

\begin{lem}\label{lem:reconf}
For every $q\in\mathcal{M}_{o}^{1}$, we have $\dfrac{\rho_{q}}{r_{q}}\leq L_{q}$ and $\Delta_{q}\,\delta_{q}\leq L_{q}$. Consequently, $\dfrac{\rho_{q}}{r_{q}},\ \Delta_{q}\,\delta_{q}\leq \min\{\tau,3\}T^{*}$.
\end{lem}
\begin{proof}
Let $v$ be a port of $q$ that sends or receives $\rho_{q}$ units of data, and let $A_{v}$ be the set of flows assigned to $q$ that are incident to $v$. Then
\begin{equation*}
\frac{\rho_{q}}{r_{q}}=\sum_{(i,j,k)\in A_{v}}\frac{d_{ijk}}{r_{q}}\leq\sum_{(i,j,k)\in A_{v}}\left(\frac{d_{ijk}}{r_{q}}+\delta_{q}\right)=\sum_{(i,j,k)\in A_{v}}p_{ijkq}\leq L_{q}.
\end{equation*}
Similarly, let $v'$ be a port of $q$ with $\Delta_{q}$ incident flows assigned to $q$, and let $A_{v'}$ be the set of these flows. Since every flow on an OCS core is charged one reconfiguration delay,
\begin{equation*}
\Delta_{q}\,\delta_{q}=\sum_{(i,j,k)\in A_{v'}}\delta_{q}\leq\sum_{(i,j,k)\in A_{v'}}p_{ijkq}\leq L_{q}.
\end{equation*}
The last statement follows from~(\ref{eq:Lq}).
\end{proof}

\begin{lem}\label{lem:AS}
If $q\in\mathcal{M}_{o}^{1}$, then $T^{s}_{q}\leq 4\min\{\tau,3\}T^{*}$.
\end{lem}
\begin{proof}
Combining~(\ref{eq:recosin}) with Lemma~\ref{lem:reconf} gives $T^{s}_{q}\leq 2(L_{q}+L_{q})=4L_{q}\leq 4\min\{\tau,3\}T^{*}$.
\end{proof}

Since $T^{*}\leq\mathrm{OPT}$ by~(\ref{eq:Tstar}) and $T^{s}=\max_{q}T^{s}_{q}$, Lemmas~\ref{lem:EPS}, \ref{lem:NAS}, and~\ref{lem:AS} yield the following theorem.

\begin{thm}\label{thm:makespan}
Algorithm~\ref{AlgMS} is a polynomial-time algorithm. Depending on the network architecture, its approximation ratio for the makespan is as follows:
\begin{itemize}
\item if all cores are EPS, the ratio is $\min\left\{\tau, 3\right\}$;
\item if all cores are not-all-stop OCS, the ratio is $2\min\left\{\tau, 3\right\}$;
\item if all cores are all-stop OCS, the ratio is $4\min\left\{\tau, 3\right\}$.
\end{itemize}
In a hybrid network, the ratio is the largest of the ratios of the architectures that are present. In particular, it is $2\min\{\tau,3\}\le 6$ when the network contains only EPS and not-all-stop OCS cores, and it is $4\min\{\tau,3\}\le 12$ whenever at least one all-stop OCS core is present.
\end{thm}

\begin{rem}\label{rem:combine}
The bounds of Theorem~\ref{thm:makespan} do not depend on $m$, whereas the bounds of~\cite{chen2025nonsplitting} grow with $m$. The case $\tau=1$ is trivial: no two flows share a port, so sending every flow to a core on which its transmission time is smallest and starting all flows at time $0$ yields an optimal makespan (on an all-stop OCS core, all circuits are set up in a single reconfiguration). For $\tau\geq 2$ and small $m$ or $\tau$, the bounds of~\cite{chen2025nonsplitting} can still be smaller. For example, when $m=2$ and $\tau\geq 3$, the bounds of~\cite{chen2025nonsplitting} are $2$ and $4$ for EPS and not-all-stop OCS, compared with our $3$ and $6$; and for all-stop OCS with $m=2$, the bound $2\min\{2\tau-1, m+\tau-1\}=2\tau+2$ is smaller than $4\min\{\tau,3\}$ for $2\leq\tau\leq 4$. Running both algorithms and returning the schedule with the smaller makespan achieves the minimum of the two guarantees in every architecture.
\end{rem}

\subsection{A Lower Bound on the Integrality Gap}
The bound $3T^{*}$ of Theorem~\ref{thm:IR} cannot be improved below $2T^{*}$ by any rounding of LP($T$), as the following example shows.

\begin{prop}\label{prop:gap}
For every $m\geq 2$, there is an instance with $m$ EPS cores for which $\mathrm{OPT}/T^{*}\geq \frac{2m}{m+1}$. Hence the integrality gap of LP($T$) is at least $2$.
\end{prop}
\begin{proof}
Consider $m$ EPS cores of rate $1$ and a single input port $i$ with $m+1$ flows of size $m$, each to a different output port. Setting $x_{ijkq}=\frac{1}{m}$ for every flow and every core gives a load of $(m+1)\cdot m\cdot\frac{1}{m}=m+1$ at port $i$ on every core, and every output port has load $1$ on every core; hence LP($m+1$) is feasible and $T^{*}\leq m+1$. In any non-splitting assignment, two of the $m+1$ flows share a core, so the load of port $i$ on that core is at least $2m$, and therefore $\mathrm{OPT}\geq 2m$. The ratio $\frac{2m}{m+1}$ tends to $2$ as $m$ grows.
\end{proof}

\section{Approximation Algorithm for Minimizing the Total Weighted Completion Time}\label{sec:Weighted}
We now turn to the total weighted coflow completion time. We first review the interval-indexed approach of~\cite{chen2025nonsplitting}, which is inspired by Chudak and Shmoys~\cite{CHUDAK1999323}, and then explain how the rounding inside each interval is replaced by Algorithm~\ref{AlgIR}.

\subsection{The Interval-Indexed Relaxation of~\cite{chen2025nonsplitting}}
As in~\cite{chen2025nonsplitting}, let $L$ be an integer, polynomial in the input size, such that the horizon $2^{L}$ is long enough to contain an optimal schedule. The time horizon is divided into the geometrically increasing intervals $[1,2], (2, 4], \ldots, (2^{L-1}, 2^{L}]$, and we set $t_{l}=2^l$ for $l=0, 1, \ldots, L$. The variable $x_{ijkql}$ indicates that flow $(i,j,k)$ is transmitted through core $q$ and completes within the $l$-th interval $(t_{l-1},t_{l}]$. The relaxation of~\cite{chen2025nonsplitting} is
\begin{subequations}\label{coflow:main:wc}
\begin{align}
&    \text{min} && \sum_{k\in \mathcal{K}}w_k C_{k}                           && \tag{\ref{coflow:main:wc}}\\
&    \text{s.t.} && \sum_{q\in \mathcal{M}} \sum_{l=1}^{L} x_{ijkql}=1, && \forall (i, j, k) \in \mathcal{F},   \label{coflow:wc:a}\\
&         && r_{k}+\sum_{q\in \mathcal{M}} p_{ijkq}\sum_{l=1}^{L}x_{ijkql}\leq C_{ijk},  && \forall (i, j, k) \in \mathcal{F}, \label{coflow:wc:b}\\
&         && \sum_{u=1}^{l} \sum_{(i,j,k)\in \mathcal{F}_{i}} p_{ijkq}\, x_{ijkqu}  \leq t_{l}, && \forall i\in \mathcal{I}, \forall q\in \mathcal{M}, \forall l\in [1, L], \label{coflow:wc:f}\\
&         && \sum_{u=1}^{l} \sum_{(i,j,k)\in \mathcal{F}_{j}} p_{ijkq}\, x_{ijkqu}  \leq t_{l}, && \forall j\in \mathcal{J}, \forall q\in \mathcal{M}, \forall l\in [1, L], \label{coflow:wc:g}\\
&         && \sum_{l=1}^{L}t_{l-1} \sum_{q\in \mathcal{M}} x_{ijkql} \leq C_{ijk}, && \forall (i, j, k) \in \mathcal{F}, \label{coflow:wc:h}\\
&         && C_{ijk} \leq C_{k}, && \forall (i, j, k) \in \mathcal{F}, \label{coflow:wc:j}\\
&         && x_{ijkql} \geq 0, && \forall (i, j, k) \in \mathcal{F}, \forall q\in \mathcal{M}, \forall l\in [1, L], \label{coflow:wc:i} \\
&         && x_{ijkql} = 0, && \text{if } r_k+p_{ijkq}>t_{l}. \label{coflow:wc:k}
\end{align}
\end{subequations}
Constraint~(\ref{coflow:wc:a}) assigns every flow to one core and one interval; (\ref{coflow:wc:b}) and~(\ref{coflow:wc:h}) are two lower bounds on the completion time of a flow, namely its release time plus its transmission time and the left end point of its interval; (\ref{coflow:wc:f}) and~(\ref{coflow:wc:g}) state that the flows completing by time $t_{l}$ load each port of each core by at most $t_{l}$; (\ref{coflow:wc:j}) says that a coflow completes no earlier than any of its flows; and~(\ref{coflow:wc:k}) excludes completions that are impossible because of the release time and the transmission time. Every non-splitting schedule induces a feasible solution with the same objective value~\cite{chen2025nonsplitting}, so the optimal value of~(\ref{coflow:main:wc}) is a lower bound on the optimum. Throughout this section, $(\textbf{x},\textbf{C})$ denotes an optimal solution of~(\ref{coflow:main:wc}), and $C_{ijk}$, $C_{k}$ refer to its values.

As in~\cite{chen2025nonsplitting}, every flow $(i,j,k)$ is assigned to the group $g(i,j,k)=\max\{g_{1},g_{2}\}$, where
\begin{equation*}
g_{1} = \min \left\{ g \ \middle| \ \sum_{l=1}^{g}\sum_{q\in \mathcal{M}} x_{ijkql} \geq \frac{1}{2} \right\}
\quad\text{and}\quad
g_{2} = \min \left\{ g \ \middle| \ C_{ijk} \leq t_{g} \right\}.
\end{equation*}
Let $\mathcal{G}_{g}=\{(i,j,k)\in\mathcal{F}\mid g(i,j,k)=g\}$. In~\cite{chen2025nonsplitting}, every flow is sent to the core that maximizes $\sum_{l=1}^{g(i,j,k)}x_{ijkql}$; since this maximum is only guaranteed to be a $\frac{1}{m}$ fraction of the total, the port loads of a group were bounded by $2m\,t_{g}$, which is the origin of the factor $m$ in the ratios of~\cite{chen2025nonsplitting}. We instead round each group by Algorithm~\ref{AlgIR}.

\subsection{Rounding Each Group by Iterative Relaxation}
For a group $g$ and a flow $(i,j,k)\in\mathcal{G}_{g}$, define
\begin{equation}\label{eq:y}
y_{ijkq}=\frac{\sum_{l=1}^{g}x_{ijkql}}{\sum_{q'\in\mathcal{M}}\sum_{l=1}^{g}x_{ijkq'l}}\qquad\text{for } q\in\mathcal{M},
\end{equation}
and let $\mathcal{S}_{g}=\{((i,j,k),q)\mid (i,j,k)\in\mathcal{G}_{g},\ r_{k}+p_{ijkq}\leq t_{g}\}$.

\begin{lem}\label{lem:groupLP}
For every group $g$, the vector $\textbf{y}$ restricted to $\mathcal{S}_{g}$ belongs to $\mathcal{P}(\mathcal{G}_{g},\mathcal{S}_{g},2t_{g})$. Moreover, $p_{ijkq}\leq t_{g}\leq 2t_{g}$ for every $((i,j,k),q)\in\mathcal{S}_{g}$.
\end{lem}
\begin{proof}
Fix $(i,j,k)\in\mathcal{G}_{g}$. Since $g\geq g_{1}$, the denominator in~(\ref{eq:y}) is at least $\frac{1}{2}$, so $\textbf{y}$ is well defined and $\sum_{q}y_{ijkq}=1$. If $y_{ijkq}>0$, then $x_{ijkql}>0$ for some $l\leq g$, and~(\ref{coflow:wc:k}) gives $r_{k}+p_{ijkq}\leq t_{l}\leq t_{g}$; hence $\textbf{y}$ is supported on $\mathcal{S}_{g}$, and $p_{ijkq}\leq t_{g}$ on $\mathcal{S}_{g}$. Finally, for an input port $i$ and a core $q$, constraint~(\ref{coflow:wc:f}) with $l=g$ yields
\begin{equation*}
\sum_{(i,j,k)\in\mathcal{F}_{i}\cap\mathcal{G}_{g}}p_{ijkq}\,y_{ijkq}\leq 2\sum_{(i,j,k)\in\mathcal{F}_{i}}p_{ijkq}\sum_{l=1}^{g}x_{ijkql}\leq 2t_{g}.
\end{equation*}
The same argument with~(\ref{coflow:wc:g}) applies to every output port.
\end{proof}

For each group $g$, we run Algorithm~\ref{AlgIR} with $\bar{\mathcal{F}}=\mathcal{G}_{g}$, $\mathcal{S}=\mathcal{S}_{g}$, and $T=2t_{g}$; by Lemma~\ref{lem:groupLP}, the input requirements are satisfied. Let $\mathcal{A}_{q,g}$ be the set of flows of $\mathcal{G}_{g}$ that are assigned to core $q$. The overall procedure is summarized in Algorithm~\ref{AlgWC}. The groups are processed one after another on every core, and group $g$ is started on core $q$ when both core $q$ has finished group $g-1$ and time $t_{g}$ has been reached. Since every flow in $\mathcal{G}_{g}$ is admissible only on cores with $r_{k}+p_{ijkq}\leq t_{g}$, all flows of group $g$ have been released when the group starts.

\begin{algorithm}
\caption{Scheduling for minimizing the total weighted completion time}
    \begin{algorithmic}[1]
				\STATE Solve the linear program (\ref{coflow:main:wc}) to obtain $(\textbf{x},\textbf{C})$
				\FOR{each flow $(i, j, k)\in \mathcal{F}$}
				    \STATE Compute $g(i,j,k)=\max\{g_{1},g_{2}\}$ and add $(i,j,k)$ to $\mathcal{G}_{g(i,j,k)}$
				\ENDFOR
				\FOR{each group $g=1,\ldots,L$ with $\mathcal{G}_{g}\neq\emptyset$}
						\STATE Compute $\textbf{y}$ by~(\ref{eq:y}) and run Algorithm~\ref{AlgIR} with $\bar{\mathcal{F}}=\mathcal{G}_{g}$, $\mathcal{S}=\mathcal{S}_{g}$, and $T=2t_{g}$
						\STATE Let $\mathcal{A}_{q,g}$ be the flows of $\mathcal{G}_{g}$ assigned to core $q$, for each $q\in\mathcal{M}$
				\ENDFOR
				\FOR{each $q\in \mathcal{M}$ in parallel}
					\FOR{$g=1,\ldots,L$}
						\STATE Wait until core $q$ has finished $\mathcal{A}_{q,g-1}$ and the current time is at least $t_{g}$
						\IF{$q\in\mathcal{M}_{e}$}
							\STATE Schedule $\mathcal{A}_{q, g}$ by the Birkhoff-von Neumann decomposition~\cite{marcus1959diagonals, Qiu2015}
						\ELSIF{$q\in\mathcal{M}_{o}^{2}$}
							\STATE Schedule $\mathcal{A}_{q, g}$ by Sunflow~\cite{huang2016}
						\ELSE
							\STATE Schedule $\mathcal{A}_{q, g}$ by Reco-Sin~\cite{Tan2021}
						\ENDIF
					\ENDFOR
				\ENDFOR
   \end{algorithmic}
\label{AlgWC}
\end{algorithm}

\subsection{Analysis}
Let $T_{i}^{g}$ and $T_{j}^{g}$ denote the load of input port $i$ and output port $j$ on a core caused by the flows of group $g$ assigned to that core; again, the index of the core is omitted.

\begin{lem}\label{lem:groupload}
For every group $g$, every core, every $i \in \mathcal{I}$, and every $j \in \mathcal{J}$, we have $T_{i}^{g}, T_{j}^{g} \leq \min\left\{\tau, 4\right\} t_{g}$.
\end{lem}
\begin{proof}
We present the argument for an input port $i$ on a core $q$; the argument for an output port is identical. Fix a group $g$ and let $c=(i,q)$.

\emph{The bound $4t_{g}$.} By Lemma~\ref{lem:groupLP}, the restriction of $\textbf{y}$ to $\mathcal{S}_{g}$ belongs to $\mathcal{P}(\mathcal{G}_{g},\mathcal{S}_{g},2t_{g})$, so the input requirement of Algorithm~\ref{AlgIR} is met with $T=2t_{g}$; condition~(\ref{eq:pmax}) also holds since $p_{e}\leq t_{g}\leq 2t_{g}$ for every $e\in\mathcal{S}_{g}$. The quantity $P$ of Theorem~\ref{thm:IR} therefore satisfies
\begin{equation*}
P=\max_{e\in\mathcal{S}_{g}}p_{e}\leq t_{g}.
\end{equation*}
We stress that this bound on $P$ is half of the capacity $T=2t_{g}$: the capacity is doubled because each flow of $\mathcal{G}_{g}$ carries only a fraction of at least $\frac{1}{2}$ of its assignment within the first $g$ intervals, but an admissible edge must still satisfy $r_{k}+p_{ijkq}\leq t_{g}$. Applying Theorem~\ref{thm:IR} to the run of Algorithm~\ref{AlgIR} for group $g$, the load of $c$ caused by the flows of $\mathcal{G}_{g}$ is at most
\begin{equation*}
T+2P\leq 2t_{g}+2t_{g}=4t_{g}.
\end{equation*}
In more detail, if $c$ is never dropped during the run for group $g$, its load is at most $2t_{g}$. Otherwise, let $\textbf{x}$ be the extreme point at the iteration in which $c$ is dropped and $\Lambda_{c}$ the load of $c$ at that moment. As in the proof of Theorem~\ref{thm:IR},
\begin{equation*}
T_{i}^{g}\leq\left(\Lambda_{c}+\sum_{e\in E(c)}p_{e}x_{e}\right)+\sum_{e\in E(c)}p_{e}(1-x_{e})\leq 2t_{g}+t_{g}\sum_{e\in E(c)}(1-x_{e})\leq 2t_{g}+2t_{g}=4t_{g},
\end{equation*}
where the first part is bounded by the residual capacity constraint~(\ref{RLP:b}) with right-hand side $2t_{g}-\Lambda_{c}$, and the second part uses $p_{e}\leq t_{g}$ and $\sum_{e\in E(c)}(1-x_{e})\leq 2$.

\emph{The bound $\tau t_{g}$.} At most $\tau$ flows are incident to port $i$. Each flow of $\mathcal{G}_{g}$ that is assigned to core $q$ is assigned through an edge of $\mathcal{S}_{g}$, so its transmission time on $q$ is at most $t_{g}$. Hence $T_{i}^{g}\leq\tau t_{g}$.

Combining the two bounds gives $T_{i}^{g}\leq\min\{\tau,4\}t_{g}$.
\end{proof}

Note that applying only the coarser bound $3T$ of Theorem~\ref{thm:IR} with $T=2t_{g}$ would give $6t_{g}$; the improvement to $4t_{g}$ comes from the fact that the deficiency part is charged at the transmission time $P\leq t_{g}$ rather than at the capacity $2t_{g}$.

Let $C_{\max}^{g}$ denote the largest time that a core spends on the flows of group $g$, measured from the moment the group is started on that core. For a core $q$, let $L_{q}^{g}$ be the largest load of a port on $q$ caused by the flows of $\mathcal{A}_{q,g}$. By Lemma~\ref{lem:groupload},
\begin{equation}\label{eq:Lqg}
L_{q}^{g}\leq\min\{\tau,4\}t_{g}\quad\text{for every } q\in\mathcal{M}.
\end{equation}

\begin{lem}\label{lem:groupms}
For every group $g$,
\begin{itemize}
\item if the cores are EPS, then $C_{\max}^{g}\leq \min\left\{\tau, 4\right\} t_{g}$;
\item if the cores are not-all-stop OCS, then $C_{\max}^{g}\leq 2 \min\left\{\tau, 4\right\} t_{g}$;
\item if the cores are all-stop OCS, then $C_{\max}^{g}\leq 4\min\left\{\tau, 4\right\}t_{g}$.
\end{itemize}
\end{lem}
\begin{proof}
Fix a group $g$ and a core $q$. Once group $g$ is started on $q$, the core schedules exactly the flows of $\mathcal{A}_{q,g}$, all of which have been released (Algorithm~\ref{AlgWC}), and no flow of another group is transmitted on $q$ until $\mathcal{A}_{q,g}$ is finished. Hence the time spent on group $g$ is the makespan of the schedule computed for the flow set $\mathcal{A}_{q,g}$ alone, and we may apply the single-core analyses of Section~\ref{sec:Makespan} to this flow set, with $L_{q}$ replaced by $L_{q}^{g}$.

\emph{EPS cores.} By Theorem~\ref{BvN2}, the Birkhoff-von Neumann based schedule completes $\mathcal{A}_{q,g}$ within the largest port load, so the time spent is at most $L_{q}^{g}\leq\min\{\tau,4\}t_{g}$ by~(\ref{eq:Lqg}).

\emph{Not-all-stop OCS cores.} By the analysis of Sunflow~\cite{huang2016}, the last flow of $\mathcal{A}_{q,g}$, say from input $i$ to output $j$, finishes no later than the sum of the loads of $i$ and $j$ caused by $\mathcal{A}_{q,g}$. Hence the time spent is at most $T_{i}^{g}+T_{j}^{g}\leq 2L_{q}^{g}\leq 2\min\{\tau,4\}t_{g}$.

\emph{All-stop OCS cores.} Let $\rho_{q}^{g}$ be the largest amount of data that a port of $q$ sends or receives in $\mathcal{A}_{q,g}$, and let $\Delta_{q}^{g}$ be the largest number of flows of $\mathcal{A}_{q,g}$ incident to one port. By~(\ref{eq:recosin}) applied to $\mathcal{A}_{q,g}$, the time spent is at most $2\left(\frac{\rho_{q}^{g}}{r_{q}}+\Delta_{q}^{g}\delta_{q}\right)$. The proof of Lemma~\ref{lem:reconf} applies verbatim to the flow set $\mathcal{A}_{q,g}$: if $v$ is a port that sends or receives $\rho_{q}^{g}$ units of data and $A_{v}$ is the set of flows of $\mathcal{A}_{q,g}$ incident to $v$, then
\begin{equation*}
\frac{\rho_{q}^{g}}{r_{q}}=\sum_{(i,j,k)\in A_{v}}\frac{d_{ijk}}{r_{q}}\leq\sum_{(i,j,k)\in A_{v}}p_{ijkq}\leq L_{q}^{g},
\end{equation*}
and if $v'$ is a port with $\Delta_{q}^{g}$ incident flows of $\mathcal{A}_{q,g}$, forming the set $A_{v'}$, then, since each of these flows is charged one reconfiguration delay,
\begin{equation*}
\Delta_{q}^{g}\delta_{q}=\sum_{(i,j,k)\in A_{v'}}\delta_{q}\leq\sum_{(i,j,k)\in A_{v'}}p_{ijkq}\leq L_{q}^{g}.
\end{equation*}
Therefore the time spent is at most $2(L_{q}^{g}+L_{q}^{g})=4L_{q}^{g}\leq 4\min\{\tau,4\}t_{g}$.

Taking the maximum over all cores of the respective type proves the lemma.
\end{proof}

Let $R$ be the factor appearing in Lemma~\ref{lem:groupms}, i.e., $R=\min\{\tau,4\}$, $R=2\min\{\tau,4\}$, and $R=4\min\{\tau,4\}$ for EPS, not-all-stop OCS, and all-stop OCS, respectively. Note that $R\geq 1$.

\begin{lem}\label{lem:finish}
On every core, all flows of groups $1,\ldots,g$ are completed by time $2R\,t_{g}$.
\end{lem}
\begin{proof}
Let $E_{g}$ be the time at which the core finishes group $g$, with $E_{0}=0$. By Algorithm~\ref{AlgWC} and Lemma~\ref{lem:groupms}, $E_{g}\leq \max\{t_{g},E_{g-1}\}+R\,t_{g}$. We show $E_{g}\leq 2R\,t_{g}$ by induction on $g$. For $g=1$, $E_{1}\leq t_{1}+R\,t_{1}\leq 2R\,t_{1}$ because $R\geq 1$. For $g\geq 2$, the induction hypothesis gives $E_{g-1}\leq 2R\,t_{g-1}=R\,t_{g}$, so $\max\{t_{g},E_{g-1}\}\leq R\,t_{g}$ and $E_{g}\leq 2R\,t_{g}$.
\end{proof}

Let $\hat{C}_{k}$ denote the completion time of coflow $k$ in the schedule produced by Algorithm~\ref{AlgWC}.

\begin{thm}\label{thm:wc}
Algorithm~\ref{AlgWC} is a polynomial-time algorithm, and for every coflow $k$,
\begin{itemize}
\item if all cores are EPS, then $\hat{C}_{k}\leq 8\min\left\{\tau, 4\right\} C_{k}$;
\item if all cores are not-all-stop OCS, then $\hat{C}_{k}\leq 16 \min\left\{\tau, 4\right\} C_{k}$;
\item if all cores are all-stop OCS, then $\hat{C}_{k}\leq 32\min\left\{\tau, 4\right\} C_{k}$.
\end{itemize}
Consequently, the approximation ratio for the total weighted coflow completion time is $8\min\{\tau, 4\}$, $16\min\{\tau, 4\}$, and $32\min\{\tau, 4\}$ in these three architectures, respectively. In a hybrid network, the ratio is the largest of the ratios of the architectures that are present, and hence at most $32\min\{\tau,4\}\leq 128$.
\end{thm}
\begin{proof}
The linear program~(\ref{coflow:main:wc}) has polynomially many variables and constraints, since $L$ is polynomial in the input size, and Algorithm~\ref{AlgIR} runs in polynomial time by Theorem~\ref{thm:IR}.

Fix a coflow $k$, and let $(i,j,k)$ be a flow of $k$ with the largest group index $g=g(i,j,k)$. By Lemma~\ref{lem:finish}, every flow of $k$ is completed by time $2R\,t_{g}$, so $\hat{C}_{k}\leq 2R\,t_{g}$. We now show $t_{g}\leq 4C_{ijk}$ and distinguish two cases.

\textbf{Case 1: $g=g_{2}$.} If $g=1$, then $t_{g}=2\leq 2C_{ijk}$, since $C_{ijk}\geq r_{k}+\min_{q}p_{ijkq}\geq 1$ by~(\ref{coflow:wc:a}) and~(\ref{coflow:wc:b}). Otherwise, the minimality of $g_{2}$ gives $t_{g-1}<C_{ijk}$, and hence $t_{g}=2t_{g-1}< 2C_{ijk}$.

\textbf{Case 2: $g=g_{1}$.} By the minimality of $g_{1}$, $\sum_{l=1}^{g-1}\sum_{q}x_{ijkql}<\frac{1}{2}$, and therefore $\sum_{l=g}^{L}\sum_{q}x_{ijkql}>\frac{1}{2}$ by~(\ref{coflow:wc:a}). Using~(\ref{coflow:wc:h}), we obtain
\begin{equation*}
\frac{t_{g-1}}{2} \leq t_{g-1}\sum_{l=g}^{L}\sum_{q\in \mathcal{M}} x_{ijkql}
\leq \sum_{l=g}^{L}\sum_{q\in \mathcal{M}} t_{l-1}\, x_{ijkql}
\leq \sum_{l=1}^{L}\sum_{q\in \mathcal{M}} t_{l-1}\, x_{ijkql}
\leq C_{ijk},
\end{equation*}
so $t_{g}=2t_{g-1}\leq 4C_{ijk}$.

In both cases $t_{g}\leq 4C_{ijk}\leq 4C_{k}$ by~(\ref{coflow:wc:j}), and thus $\hat{C}_{k}\leq 8R\,C_{k}$. Substituting the three values of $R$ proves the per-coflow bounds. Multiplying by $w_{k}$ and summing over $k$ gives $\sum_{k}w_{k}\hat{C}_{k}\leq 8R\sum_{k}w_{k}C_{k}$, and the right-hand side is at most $8R$ times the optimum, because~(\ref{coflow:main:wc}) is a relaxation. In a hybrid network, each core satisfies Lemma~\ref{lem:groupms} with the factor of its own architecture; taking $R$ as the largest factor among the architectures present, Lemma~\ref{lem:finish} and the argument above hold verbatim.
\end{proof}

As in Remark~\ref{rem:combine}, the guarantees of Theorem~\ref{thm:wc} and those of~\cite{chen2025nonsplitting} can be combined by computing both schedules and keeping the one with the smaller total weighted completion time.

\section{Concluding Remarks}\label{sec:Conclusion}
This paper revisited non-splitting coflow scheduling in heterogeneous parallel networks composed of EPS, not-all-stop OCS, and all-stop OCS cores. The previous algorithms for this problem~\cite{chen2025nonsplitting} have approximation ratios that grow linearly with the number of cores $m$ or with the maximum flow degree $\tau$, because rounding each flow to its largest fractional core may multiply the load of a port by $m$. We replaced this step by an iterative relaxation rounding that controls the input-port and the output-port capacity constraints at the same time, and we proved that the load of every port on every core is at most $\min\{\tau,3\}T^{*}$. For all-stop OCS, we additionally showed that the number of reconfigurations at a port is controlled by its load, which removes the additive $\tau$ term caused by reconfiguration delays. For minimizing the makespan, the resulting approximation ratios are $\min\{\tau,3\}$, $2\min\{\tau,3\}$, and $4\min\{\tau,3\}$ for pure EPS, pure not-all-stop OCS, and pure all-stop OCS networks. For minimizing the total weighted coflow completion time, applying the same rounding within every interval of the interval-indexed relaxation yields approximation ratios of $8\min\{\tau,4\}$, $16\min\{\tau,4\}$, and $32\min\{\tau,4\}$. In a hybrid network, the guarantee is determined by the weakest architecture present, so the overall ratios are at most $12$ and $128$, respectively. To the best of our knowledge, these are the first approximation ratios for non-splitting coflow scheduling in heterogeneous parallel networks that do not depend on the number of network cores.

Several questions remain open. First, Proposition~\ref{prop:gap} shows that the integrality gap of LP($T$) is at least $2$, while our rounding achieves a port-load bound of $3T^{*}$; closing this gap, for example with a refined rule for dropping constraints or with a strengthened relaxation, would directly improve all of our ratios. Second, the constants for the total weighted completion time come from several sources, including the group threshold $\frac{1}{2}$, the geometric interval length, and the factor $2$ in the capacity of each group, and it would be interesting to reduce them, for example by randomizing the interval boundaries. Third, the schedules on individual OCS cores lose a factor of $2$ or $4$ relative to the port loads; designing assignment and scheduling procedures that are analyzed jointly may reduce this loss. Finally, extending the results to online settings and to coflows with precedence constraints is a natural direction for future work.

\section*{Declaration of Generative AI Use}
During the preparation of this work, the author used generative AI tools to assist with language editing, literature organization, and discussions of parts of the algorithmic analysis. All research ideas, algorithm designs, and proofs were checked and verified by the author, who takes full responsibility for the content of this paper.

\appendix
\renewcommand{\thesection}{A}\setcounter{section}{1}
\section*{Appendix: Birkhoff-von Neumann Decomposition}\label{sec:BnV}
This appendix recalls why the flows assigned to a single EPS core can be completed within the largest port load of that core. Since only one core is involved, we drop the core index and describe the assigned flows by an $N\times N$ matrix $D=\left(d_{i,j}\right)_{i,j=1}^{N}$ of transmission times, and we define its load as
\begin{equation*}
\rho(D)=\max \left\{\max_{i}\sum_{j'=1}^{N}d_{i,j'},\ \max_{j}\sum_{i'=1}^{N}d_{i',j}\right\}.
\end{equation*}
The procedure relies on the classical theorem of Birkhoff and von Neumann~\cite{marcus1959diagonals}.

\begin{thm}\label{BvN}
\cite{marcus1959diagonals} \textbf{(Birkhoff-von Neumann)} Every doubly stochastic matrix $A\in \mathbb{R}^{N\times N}$ can be written as $A=\sum_{i=1}^{h} c_{i}P_{i}$, where each $P_{i}$ is a permutation matrix, $c_{i}\in (0, 1]$, $\sum_{i=1}^{h} c_{i}=1$, and $h\leq N^2-2N+2$.
\end{thm}

Algorithm~\ref{BvN_alg}, due to Qiu \textit{et al.}~\cite{Qiu2015}, works in two phases. The first phase repeatedly increases an entry lying in a row and a column whose sums are below $\rho(D)$, until every row sum and every column sum of the resulting matrix $\widetilde{D}\geq D$ equals $\rho(D)$. Then $\widetilde{D}/\rho(D)$ is doubly stochastic, and by Theorem~\ref{BvN}, $\widetilde{D}=\sum_{u}q_{u}\Pi_{u}$ for permutation matrices $\Pi_{u}$ and positive coefficients with $\sum_{u}q_{u}=\rho(D)$. The second phase extracts these permutation matrices one at a time: it finds a perfect matching in the bipartite graph formed by the positive entries of $\widetilde{D}$, transmits along the matched input-output pairs for the smallest matched entry $q$ time units, and subtracts $q$ times the matching from $\widetilde{D}$.

\begin{algorithm}
\caption{Birkhoff-von Neumann based schedule~\cite{Qiu2015}}
    \begin{algorithmic}[1]
				\STATE $\widetilde{D}\leftarrow D$
				\WHILE{some row sum or column sum of $\widetilde{D}$ is smaller than $\rho(D)$}
						\STATE Let $i^{*}$ be a row with the smallest row sum and $j^{*}$ a column with the smallest column sum of $\widetilde{D}$
						\STATE $p\leftarrow\rho(D)-\max\left\{\sum_{j'}\widetilde{d}_{i^{*},j'},\ \sum_{i'}\widetilde{d}_{i',j^{*}}\right\}$
						\STATE $\widetilde{d}_{i^{*},j^{*}}\leftarrow\widetilde{d}_{i^{*},j^{*}}+p$
				\ENDWHILE
				\WHILE{$\widetilde{D}\neq 0$}
						\STATE Find a perfect matching $M$ in the bipartite graph with an edge $(i,j)$ whenever $\widetilde{d}_{i,j}>0$
						\STATE $q\leftarrow\min\{\widetilde{d}_{i,j}\mid (i,j)\in M\}$
						\STATE Transmit from $i$ to $j$ for $q$ time units for every $(i,j)\in M$ with remaining demand
						\STATE $\widetilde{d}_{i,j}\leftarrow\widetilde{d}_{i,j}-q$ for every $(i,j)\in M$
				\ENDWHILE
   \end{algorithmic}
\label{BvN_alg}
\end{algorithm}

In each iteration of the first loop, at least one row or column reaches the sum $\rho(D)$, and in each iteration of the second loop, at least one positive entry becomes zero; hence both loops terminate after polynomially many iterations. The total transmission time equals $\sum_{u}q_{u}=\rho(D)$, which gives the following theorem.

\begin{thm}\label{BvN2}
Algorithm~\ref{BvN_alg} runs in polynomial time and completes all flows described by $D$ within $\rho(D)$ time units.
\end{thm}

\end{document}